\documentclass{article}
\usepackage{epsf}
\usepackage{amssymb,ComplexSystems}
\usepackage{color}
\usepackage{amsmath}
\usepackage{graphicx}
\usepackage{subcaption}
\usepackage{tikz}
\usepackage{url}

\definecolor{sequentialgrovercolor}{RGB}{73, 190, 170}
\definecolor{parallelgrovercolor}{RGB}{239, 118, 122}
\definecolor{matplotlibgray}{RGB}{127, 127, 127}

\definecolor{matplotlibgreenk31}{RGB}{31, 119, 180}
\definecolor{matplotlibgreenk32}{RGB}{255, 127, 14}
\definecolor{matplotlibgreenk33}{RGB}{44, 160, 44}
\definecolor{matplotlibgreenk34}{RGB}{214, 39, 40}
\definecolor{matplotlibgreenk35}{RGB}{148, 103, 189}

\newcommand{\ket}[1]{|#1\rangle}
\newcommand{\tikzcircle}[2][red,fill=red]{\tikz[baseline=-0.5ex]\draw[#1,radius=#2] (0,0) circle ;}%
\newcommand{\legendmarker}[1]{\tikzcircle[#1, fill=#1]{3pt}}

\DeclareMathOperator{\spn}{span}

\newcommand{\CNOT}{CNOT}
\newcommand{\IAM}{IAM}
\newcommand{\oracle}{O}
\newcommand{\PGseq}[1]{PG_{#1\text{-seq}}}
\newcommand{\PGpar}[1]{PG_{#1\text{-par}}}

\newcommand{\itgrover}{t_G}

\newcommand*{\mathoverlineplot}[1]{%
  \overline{\mbox{\textit{#1}\rule{0pt}{3mm}}}
}
\newcommand*{\mathoverline}[1]{%
  \overline{\vphantom{F}#1}
}

\newcommand*{\mathoverlinee}[1]{%
  \overline{\vphantom{F}#1}
}

\begin{document}

\title{The parallel Grover as dynamic system%
%
}

\author{\authname{Alexander Goscinski}\\[2pt] 
\authadd{State Key Laboratory of Advanced Optical Communication Systems and }\\
\authadd{Networks, Key Lab on Navigation and Location-based Service,}\\
\authadd{Department of Electronic Engineering, Shanghai Jiaotong University,}\\
\authadd{Dongchuan Road 800, Shanghai 200240, China}
}

%
%
\markboth{Complex Systems} 
{The parallel Grover as dynamic system} 

\maketitle

\begin{abstract}
A sequential application of the Grover algorithm to solve the iterated search problem has been improved by Ozhigov \cite{ozhigov1999speedup} by parallelizing the application of the oracle. In this work a representation of the parallel Grover as dynamic system of inversion about the mean and Grover operators is given. Within this representation the parallel Grover for $k=2$ can be interpreted as rotation in three-dimensional space and it can be shown that the sole application of the parallel Grover operator does not lead to a solution for $k>2$. We propose a solution for $k=3$ with a number of approximately $1.51\sqrt{N}$ iterations. 
\end{abstract}


\section{Introduction}
Farhi and Gutmann presented an algorithm for the iterated search problem for $k=2$ \cite{farhi1997quantum}. The algorithm is a sequential application of the Grover operator with the two given oracles. First, the Grover operator with the $f_1$ oracle is applied $[\pi\sqrt{N}/4]$ times, then the Grover operator with the $f_2$ oracle is applied $[\pi\sqrt{N}/4]$ times. The complexity is therefore $[2\pi\sqrt{N}/4]$. Ozhigov was able to show that, by executing the two oracles in parallel, a speed up by a constant factor of $\sqrt{2}$ is possible \cite{ozhigov1999speedup}. Even though the speed up is negligible small, he showed that there exists a method of parallelization in the quantum circuit model beyond classical methods for parallelization. Ozhigov mainly analysis the effects of the parallel Grover for two oracles. He gives a generalization for higher $k$'s to explain that no significant speedup can be obtained with his method for higher $k$'s. In this paper we give a different approach to the problem describing the parallel Grover as dynamic system of local inversion about the mean and Grover operators. Within this representation we can give a geometric interperation of the parallel Grover for $k=2$. Furthermore, we introduce an approximation of the parallel Grover for higher $k$'s. With this approximation we can conclude that a simple periodic application of the parallel Grover operator does not solve the iterative search problem with negligible error. We introduce the problem of diversion of the amplitude to the solution state and give a solution for $k=3$.

\subsection{Iterated search problem}
The standard search problem considers one oracle with a unique solution.
The iterated search problem (ISP) considers multiple oracles for different parts of the input.
Given an input $x$ partitioned into $k$ equal-sized substrings of $x$, such that $x_1x_2\cdots x_k = x$, and multiple oracles of the form $f_i(x_1x_2\cdots x_i)$ with $i = 1,\ldots,k$, then $k$-ISP is the problem of finding the unique solution.
We assume that every oracle has a unique solution. More formal, let $e_1\cdots e_k$ be the unique solution to our problem, then each oracle is defined for $i=1,\ldots,k$ as
\[f_i(x_1x_2\cdots x_i) = \begin{cases} 1 & \text{if }x_1\cdots x_i = e_1\cdots e_i \\ 0 & \text{else.}\end{cases}\]
In this work we use the convention that each substring $x_i$ represents $n$ qubits and the number of possibilities for $x_i$ is $N=2^n$.
We use the following naming convention to separate the state spaces: $\ket{e_i}$ represents the solution at the $i$th position and $\ket{N_i} = \sum_{x=0, x\neq e_i}^{N-1}\ket{x}$ represents the remaining states at the $i$th position.

We define the set of strings $S_j = \{e_j,N_j\}$,
\[T_{i,j}=\begin{cases}S_i\times\cdots\times S_j &\text{, if } i\leq j \\ \emptyset&\text{, else}\end{cases}\]
and $T_{j} = T_{1,j}$.
We define the normalized states as $\mathoverlinee{S_j} = \{\mathoverlinee{e_j},\mathoverlinee{N_j}\}$,
\mbox{$\mathoverlinee{T_{i,j}}=$} \mbox{$\{\mathoverlinee{s}|s\in T_{i,j}\}$} and
the state $\ket{\mathoverline{T_k}} = \prod_{\overline{s}\in \mathoverlinee{T_k}}\ket{\mathoverlinee{s}}$ 
Further, we define the space of valid quantum states in $W_j=\spn(\{\ket{s} |s\in T_j\})$ as
\[Q_j = \{\sum_{\overline{s}\in \mathoverlinee{T_j}}a_{\overline{s}}\ket{\mathoverlinee{s}}\in W_j| \sum_{\overline{s}\in \overline{T_j}}a_{\overline{s}}^2 = 1\}.\] 
Given a state $\ket{s}\in Q_k$, then the state in normalized notation $\ket{\mathoverlinee{s}}$ can be expressed as
\[a\ket{s} = a(N-1)^{\#_N(s)/2}\ket{\mathoverlinee{s}}, \]
where $\#_N(s)$ is the number of $\ket{N_j}$ substates in $\ket{s}$. For example for $k=2$
\begin{align*}
&a\ket{e_1e_2} = a\ket{\mathoverline{e_1e_2}}, &a\ket{e_1N_2} = a\sqrt{N-1}\ket{\mathoverline{e_1N_2}},\\
&a\ket{N_1N_2} = a(N-1)\ket{\mathoverline{N_1N_2}}, &a\ket{N_1e_2} = a\sqrt{N-1}\ket{\mathoverline{N_1e_2}}.
\end{align*}
We call the state $\ket{N_1\ldots N_k}$ source state, $\ket{e_1\ldots e_k}$ sink state and the path $(\ket{N_1\ldots N_k},$ $\ket{e_1N_2\ldots N_k},$\ldots, $\ket{e_1\ldots e_k})$ will be called main path. 

\subsection{Wire notation}
We use strings to describe the wires in the circuit model.
Each character describes one wire, thus substrings describe multiple wires.
Strings are written in front of the operators as input in brackets.
The purpose is to be more accurate with the mathematical description of the circuit. 
As an example we give a description of an oracle for the $2$-ISP.
Let $O_2$ be the gate or matrix representing a decision function $f:\{0,\ldots, 2^{2n-1}\} \rightarrow \{0,1\}$.
Since the function has an input space of $2^{2n}$, $2n$ qubits are required to describe the input of the function.
We will describe these inputs with the subsstring $x_1x_2$ with $x_1 = x_{1_1}\ldots x_{1_n}$ and $x_2 = x_{2_1}\ldots x_{2_n}$.
Usually, to build a quantum gate out of a general function $f$, ancilla qubits are required.
Furthermore, the output wire of the qubit gate is also required to be known.
However, since the ancilla qubits as well as the exact architecture of the quantum gates are of not interest for this work, this information will be ignored in our notation.
The resulting gate of this function is described as $O_2(x_1x_2)$.
Furthermore, for the $\CNOT$ operator the wire(s) represented by the first string are the control qubit(s), and by the second string are the target qubit(s) of the $\CNOT$ operation.
Both strings will be separated by a comma (e.g. $\CNOT(x_1,x_3)$).
A $CNOT$ operator with substrings as input is defined as multiple $CNOT$ gates $CNOT(x_1,x_3) = \prod_{i=1}^n CNOT(x_{1_i},x_{3_i})$.
When a collection of gates can be executed in parallel, we put them into square brackets.
An example of this formalism can be seen in equation \ref{eq:PG_parallel_2} with corresponding circuit diagram in Fig.~\,\ref{eq:PG_parallel_2}.

\subsection{Grover operator as rotation}
The Grover operator can be geometrically interpreted as a rotation by approximately $2/\sqrt{N}$ radians in the subspace spanned by $\ket{e_1}$ and $\ket{N_1}$ \cite{aharonov1999quantum} and can be expressed in matrix form as
\begin{equation}
G (a\ket{\mathoverline{e_1}}+b\ket{\mathoverline{N_1}}) \doteq \begin{bmatrix} (1-\frac2{N}) & 2\frac{\sqrt{N-1}}{N} \\ -2\frac{\sqrt{N-1}}{N} & (1-\frac2{N}) \end{bmatrix} \begin{bmatrix} a \\ b \end{bmatrix}.
\nonumber
\end{equation}
The symbol $\doteq$ stands for "represented by" as it was introduced by \cite[p.~20]{sakurai1995modern}.
We define $\itgrover(N)=\pi\sqrt{N}/4$ as the number of iterations for the Grover algorithm to achieve a success probability with negligible error for large $N$ \cite{boyer1996tight}.
The implementation of the Grover operator consists of an oracle and an inversion about the mean (IAM) operator. We can decompose the rotation operator into two operators: 
\[ G (a\ket{\mathoverline{e_1}}+b\ket{\mathoverline{N_1}}) \doteq \underbrace{\begin{bmatrix} (-1) (1-\frac2{N}) & 2\frac{\sqrt{N-1}}{N} \\ 2\frac{\sqrt{N-1}}{N} & (1-\frac2{N}) \end{bmatrix}}_{\text{IAM operator}} 
\underbrace{\large\begin{bmatrix}-1&0\\0&1\end{bmatrix}}_{\text{Oracle operator}} \begin{bmatrix} a \\ b \end{bmatrix}.\]
The oracle operator can be interpreted as a reflection about the $\ket{e_1}$ space, therefore the IAM operator can be interpreted as a reflection about the $\ket{e_1}$ space with a subsequent rotation by the Grover operation. Both operations are orthogonal transformations. We describe the oracle operator with $I_1$.
\subsection{Sequential Grover operator}
The idea of the sequential Grover operator is to apply $k$ different Grover operator combinations of oracle and IAM operators $[\itgrover(N)]$ times. By this procedure the amplitude is transferred sequentially through the states $\ket{N_1\ldots N_k}, \ket{e_1N_2\ldots N_k}, \ldots, \ket{e_1\ldots e_k}$. The first $[\itgrover(N)]$ steps the amplitude is transported from $\ket{N_1\ldots N_k},$ to $\ket{e_1N_2\ldots N_k}$, in the next $[\itgrover(N)]$ steps from $\ket{e_1N_2\ldots N_k}$ to $\ket{e_1e_2N_3\ldots N_k}$ and so on. Thus, the overall number of iterations is $[k\itgrover(N)]$. The circuit can be expressed as follows:
\begin{multline}
\prod_{i=1}^k(SG_{k+1-i}(x_1\cdots x_i))^{\itgrover(N)}\\\text{ with }SG_{i}(x_1\cdots x_i) = \IAM(x_{i}) \oracle_{i}(x_1\cdots x_{i}) .
\nonumber
\end{multline}
For example for $k=2$ the sequential Grover is $(IAM(x_2)O_2(x_1x_2))^{\itgrover(N)}$ $(IAM(x_1)O_1(x_1))^{\itgrover(N)}$.
\section{Main results}
In this section we give the interpretation of the parallel Grover operator for $k=2$ as rotation in three-dimensional space and discuss why the sole application of the $PG_k$ does not lead to a solution with negligible error for higher $k$'s within $kt_G$ iterations. Before discussing the problem, we proof that we can approximate the IAM operators in the parallel Grover operator with an composition of reflections. Then we give a possible solution for $k=3$. Several steps in this section have been calculated with SymPy \cite{sympy} and can be found in \cite{code_results}.
\subsection{$PG_2$ as rotation within a 3-sphere}
\label{sec:parallel_grover_operator_for_k=2}
The parallel Grover as Ozhigov described it in \cite{ozhigov1999speedup} solves the $2$-ISP. Given the two oracle functions $f_1$ and $f_2$, we ommit the ancilla qubits required for the oracle gates $\oracle_1$ and $\oracle_2$ representing the oracle functions. The wires are splitted in $3$ sets of $n$ wires represented by the strings $x_1$, $x_2$ and $x_3$. Then the parallel Grover operator in its parallel form $\PGpar{2}$ is defined as the follows:
\begin{multline}
\PGpar{2}(x_1x_2x_3) = [\IAM(x_1)\IAM(x_2)] \\ \CNOT(x_3,x_1)[\oracle_1(x_3)\oracle_2(x_1x_2)]\CNOT(x_1,x_3) .
\label{eq:PG_parallel_2}
\end{multline}
In Fig.~\,\ref{fig:PG_2-parallel} we can see the circuit diagram of equation \ref{eq:PG_parallel_2}.
\begin{figure}
\centerline{\includegraphics[scale=1.]{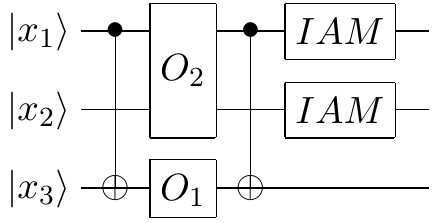}}
\caption{The circuit diagram representing $\PGpar{2}$ in equation \ref{eq:PG_parallel_2}.}
\label{fig:PG_2-parallel}
\end{figure}
In the parallel form it is not clear how the parallel Grover works, therefore Ozhigov reformulated the operator into a sequential form to simplify the analysis:
\[\PGseq{2}(x_1x_2) = \underbrace{\IAM(x_1)O_1(x_1)}_{= PG_{2_1}(x_1)} \underbrace{\IAM(x_2)\oracle_2(x_1x_2)}_{= PG_{2_2}(x_1x_2)} .\]
We can see that it is composed of two operators. These two operators are composed of an oracle and an IAM operator. When applying the operator on an arbitrary quantum state of the form $\ket{s} = a_1|e_1e_2\rangle + a_2|N_1e_2\rangle + a_3|e_1N_2\rangle + a_4|N_1N_2\rangle$, each operator $PG_{2_2}$ and $PG_{2_1}$ can be separated into two operations applied in parallel each on two subspaces.
\begin{equation}
\begin{split}
&PG_{2_2}\ket{s} = |e_1\rangle G(a_1|e_2\rangle + a_3|N_2\rangle) \\ &\hspace{8em} + |N_1\rangle \IAM(a_2|e_2\rangle + a_4|N_2\rangle)\\ 
&PG_{2_1}\ket{s} = G(a_1|e_1\rangle + a_2|N_1\rangle)|e_2\rangle \\ &\hspace{8em} + G(a_3|e_1\rangle + a_4|N_1\rangle)|N_2\rangle .
\label{eq:PG_seq_2_matrix_representation}
\end{split}
\end{equation}
For simplification we write operations like $PG_{2_2}$ as $\big[G(e_1e_2,e_1N_2)$ $\IAM(e_1N_2,N_1N_2)\big]$. Additional, a reflection of a state $\ket{s}$ will be expressed as $I_1(s)$.
We express the parallel Grover in form of a graph of Grover and IAM operations. The states represent nodes and the edges represent operations on the connecting states. 
The edge operation is the operation applied on the two states. The direction of the edge determines the upper space of the operation in matrix form. The order of the operation is determined by the number at the edge. The resulting graph for $PG_2$ can be seen in Fig.~\,\ref{fig:k2_operator_graph} together with an approximation which can be obtained with the results from Theorem~\,\ref{th:replace_theorem} and \ref{th:cubic_iam} in a subsequent section.
\begin{figure}
\centerline{
\begin{subfigure}[b]{0.35\textwidth}
    \includegraphics[width=\linewidth]{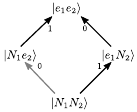}
    \caption{Operator graph $PG_2$.}
    \label{fig:k2_operator_graph}
\end{subfigure}    
\hspace{5em}
\begin{subfigure}[b]{0.35\textwidth}
    \includegraphics[width=\linewidth]{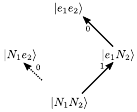}
    \caption{$PG_2$ approximation.}
    \label{fig:k2_operator_graph_approx}
\end{subfigure}
}\vspace{1em}
\centerline{
\begin{subfigure}[b]{0.43\textwidth}
    \includegraphics[width=\linewidth]{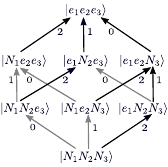}
    \caption{Operator graph $PG_3$.}
    \label{fig:k3_operator_graph}
\end{subfigure}
\hspace{3em}
\begin{subfigure}[b]{0.43\textwidth}
    \includegraphics[width=\linewidth]{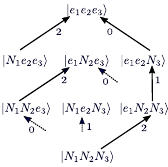}
    \caption{$PG_3$ approximation.}
    \label{fig:k3_operator_graph_approx}
\end{subfigure}
}
\caption{Operator graphs of different operators. \mbox{\raisebox{0.25ex}{\legendmarker{matplotlibgray}} IAM operator}, \mbox{\raisebox{0.25ex}{\legendmarker{black}} Grover operator} and the dotted vector represents $I_1$.}
\label{fig:operator_graphs}
\end{figure}
We made an alternative proof of Ozhigov's results (see Appendix~\ref{sec:appendix_a}), which we believe is more accessible than the original proof. Our results show that $PG_2$ applied on the initial state $\ket{\mathoverline{N_1N_2}} + O(1/\sqrt{N})\ket{\mathoverline{T_2}}$ can be interpreted as rotation in a sphere within $\spn(\{\ket{e_1e_2}, \ket{N_1e_2}, \ket{N_1N_2}\})$. The same interpretation can also be applied for the sequential Grover.
The corresponding rotations expressed with the Euler-Rodriguez formula in quaternion form are
\begin{align*}
SG_{1}^{c\sqrt{N}}\doteq\, & \cos(c) + \sin(c)\ket{\mathoverline{e_1e_2}}\\
SG_{2}^{c\sqrt{N}}\doteq\, & \cos(c) + \sin(c)\ket{\mathoverline{N_1N_2}}\\
PG_2^{c\sqrt{N}}\doteq\, & \cos(\sqrt{2}c) + \frac1{\sqrt{2}} \sin(\sqrt{2}c)(\ket{\mathoverline{e_1e_2}} + \ket{\mathoverline{N_1N_2}}).
\end{align*}
The rotation of the $\ket{\mathoverline{N_1N_2}}$ vector by the parallel and the sequential Grover is visualized in Fig.~\,\ref{fig:k2_sphere}.
\subsection{Diverted amplitude problem}
The circuit for $\PGseq{k}$ in recursive form is
\[\PGseq{k}(x_1\cdots x_k) =\\ \PGseq{k-1}(x_1\cdots x_{k-1})IAM(x_k)O_k(x_1\cdots x_k)\]
With methods we explain in the next section $PG_3$ can be approximated as it can be seen in Fig.~\,\ref{fig:k3_operator_graph}.
\begin{figure}
\centerline{
\begin{subfigure}[b]{0.44\textwidth}
    \includegraphics[width=\linewidth]{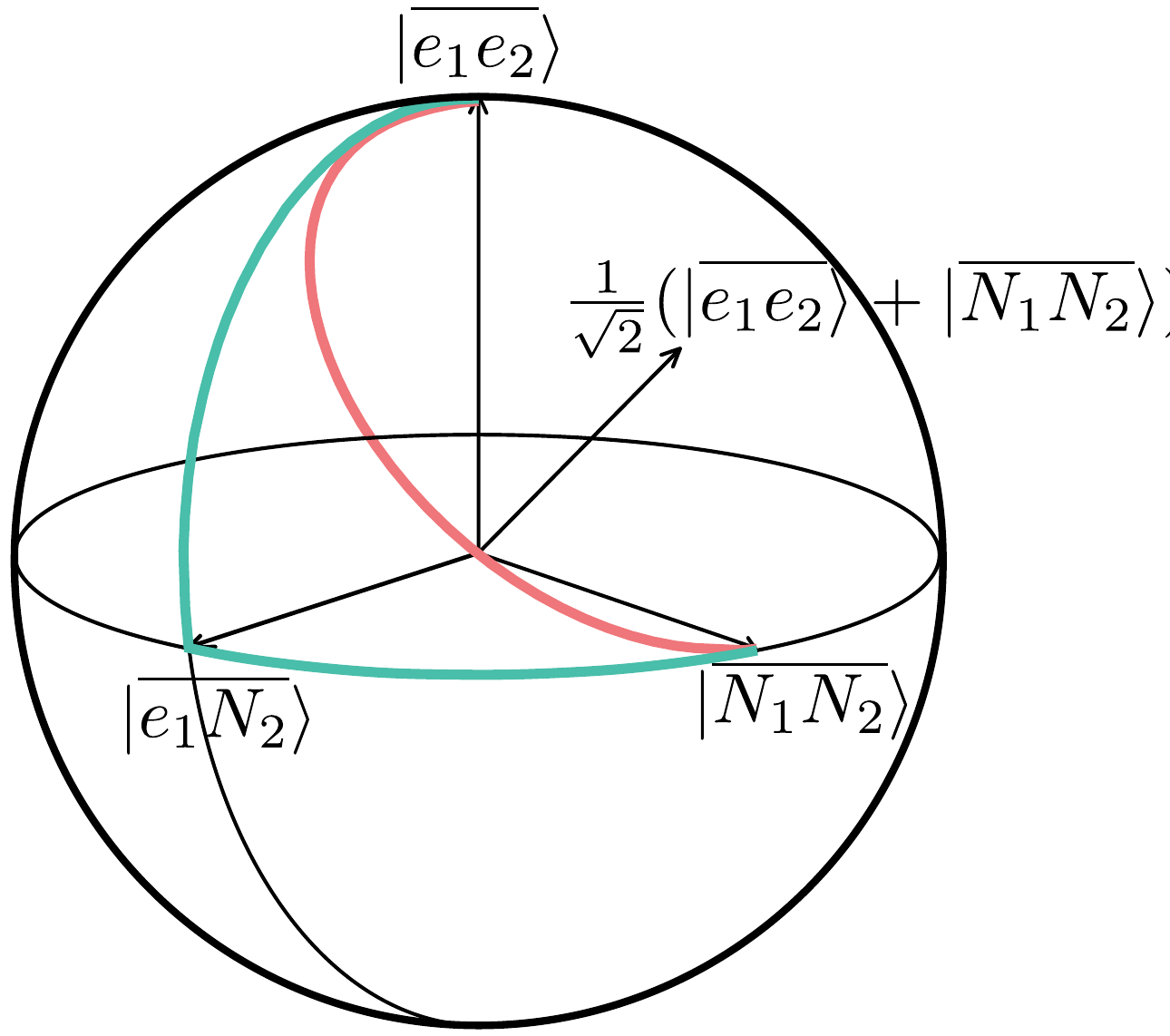}
    \vspace{0.2em}
    \caption{Solution for $2$-ISP as rotation in a sphere. \mbox{\raisebox{0.25ex}{\legendmarker{sequentialgrovercolor}} Sequential Grover}, \mbox{\raisebox{0.25ex}{\legendmarker{parallelgrovercolor}} Parallel Grover}.}
    \label{fig:k2_sphere}
\end{subfigure}
\hspace{0.5em}
\begin{subfigure}[b]{0.548\textwidth}
    \includegraphics[width=\linewidth]{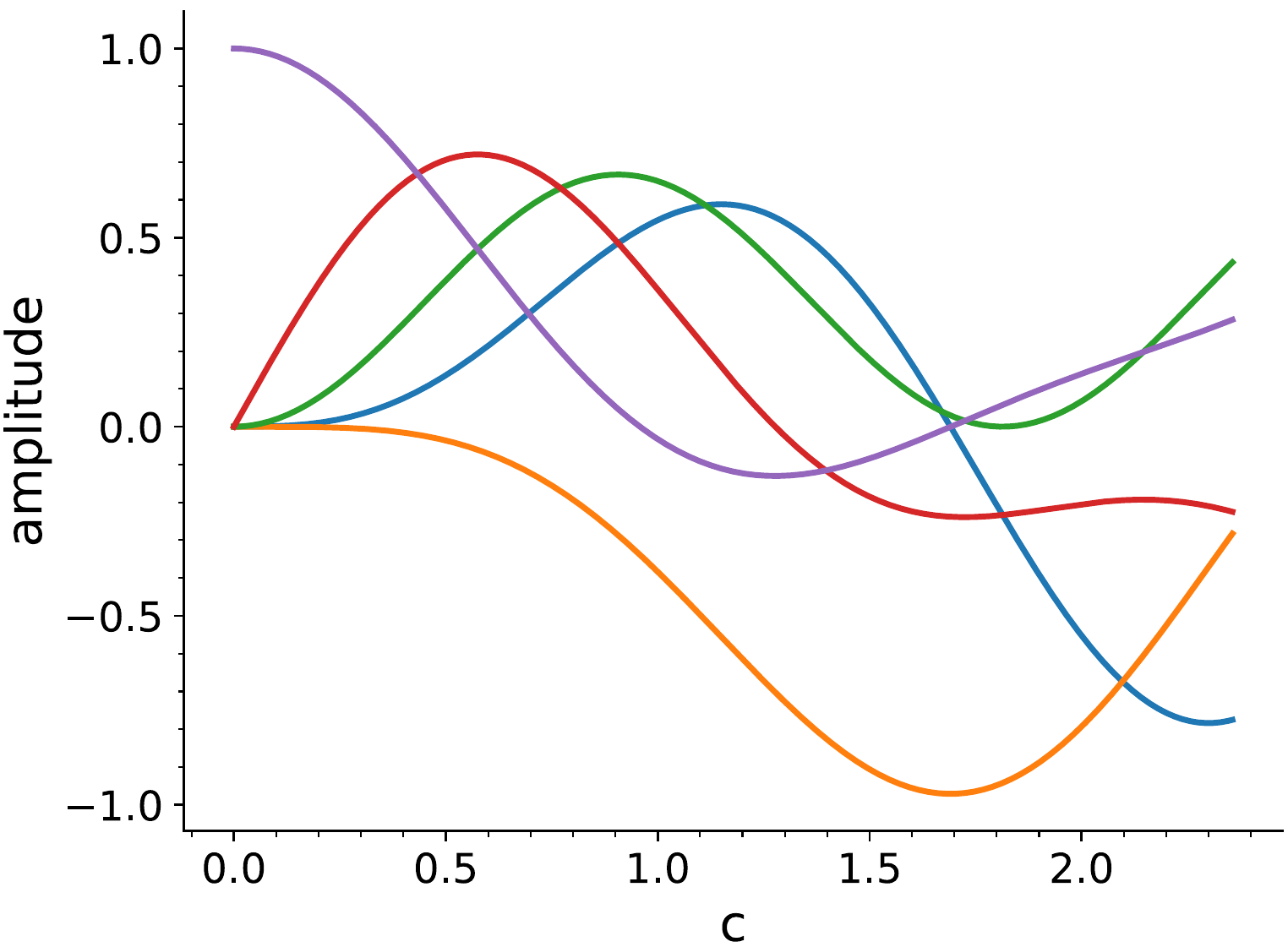}
    \caption{The amplitude evolution of $PG_3^{c\sqrt{N}}\ket{N_1N_2N_3}$. \mbox{\raisebox{0.25ex}{\legendmarker{matplotlibgreenk31}} $\ket{e_1e_2e_3}$}, \mbox{\raisebox{0.25ex}{\legendmarker{matplotlibgreenk32}} $\ket{N_1e_2e_3}$}, \mbox{\raisebox{0.25ex}{\legendmarker{matplotlibgreenk33}} $\ket{e_1e_2N_3}$}, \mbox{\raisebox{0.25ex}{\legendmarker{matplotlibgreenk34}} $\ket{e_1N_2N_3}$}, \mbox{\raisebox{0.25ex}{\legendmarker{matplotlibgreenk35}} $\ket{N_1N_2N_3}$}.}
\label{fig:k3_amplitude_evolution}
\end{subfigure}
}
\caption{Visualization of $PG_2$ and $PG_3$.}
\label{fig:k3_operator_graph_amplitude_evolution}
\end{figure}
Due to the operation $G(e_1e_2e_3,e_1e_2N_3)$ the amplitude is diverted from the state $\ket{e_3e_2e_1}$ to $\ket{e_1e_2N_3}$ as it can be seen in the amplitude evolution in Fig.~\,\ref{fig:k3_amplitude_evolution}.
The reason why the edge operations $G(e_1e_2e_3,e_1N_2e_3)$ and $G(e_1e_2N_3,N_1e_2N_3)$ can be removed but for $G(e_1e_2e_3,N_1e_2e_3)$ not depends on the number of incident IAM operators. For $G(e_1e_2e_3,N_1e_2e_3)$ the number is even while for the other cases it is uneven. We proof this in the next section. As a consequence the diversion of amplitude occurs for higher $k$'s more significant. We propose a solution to the problem in a subsequent section for the $3$-ISP.

\subsection{Approximation of $PG_k$}
For a state $\ket{s}\in Q_k$ depending on the wire the IAM gate is applied on,
the application of the IAM gate is the same as applying multiple local IAM operations
\[\IAM(x_i)\ket{s} = \prod_{g\in T_{i-1},h\in T_{i+1,k}} \IAM(\ket{ge_ih},\ket{gN_ih})\]
For simplification we will write it as $\IAM(x_i)\ket{s} = \IAM(T_{i-1}e_iT_{i+1,k},$ $T_{i-1}N_iT_{i+1,k})$.
Depending on the preceding oracle operator even a Grover or an IAM operator is applied on the pairs of states.
\begin{theorem}
Let $S:V\rightarrow V$, $A,B:U\rightarrow U$ be linear operators with $U\subseteq V$. Assume $((S-I)v)_{|U} = O(1/\sqrt{N})$, $Av = Bv + O(1/\sqrt{N})$ for any vector $\|v\|\leq 1$ and $A^2,B^2 = I + O(1/N)$, then $A$ can be replaced with $B$ within $(AS)^n$ with a negligible error $(AS)^nv = (BS)^nv + O(1/\sqrt{N})$ for a natural number $n$ in $O(\sqrt{N})$.
\label{th:replace_theorem}
\end{theorem}
\begin{proof}
Let $r_i$ be the change of $S$ in $U$ in the $i$th step with $\|r_i\| = O(1/\sqrt{N})$,
then we can express $(AS)^iv = v_{i}$ recursively
\[v_{i+1} = A(v_i+r_{i+1}).\]
Applying Equation~\ref{eq:approx_nthpower} with the substitution $A^2= I + O(1/N)$, we can approximate $A^i$ for a natural number $i$ in $O(\sqrt{N})$ 
\[A^i = \begin{cases} A + O(1/\sqrt{N}) & i\text{ uneven}\\ I + O(1/\sqrt{N}) & i\text{ even.}\end{cases}\]
Applying this on the explicit form of $v_i$ the total error propagates like $\sum_{j=1}^i \sum_{l=1}^j \frac1N = O(1/\sqrt{N})$,
thus we can approximate $v_i$ with an error in $O(1/\sqrt{N})$
\[v_i = A^{i+1 \bmod 2}\bigg(v_0 +\sum_{j\text{ uneven}}^i r_{j} \bigg) + A^{i \bmod 2}\sum_{j\text{ even}}^i r_{j} + O(\frac{1}{\sqrt{N}}).\]
The same steps can be done for $SB^i$. Since $Ar_j= Br_j + O(1/N)$ and $Av_0 = v_0 + O(1/\sqrt{N})$, we can replace $A$ with $B$ 
within $(SA)^n$ up to an error in $O(1/\sqrt{N})$.
\end{proof}
Subsystems only connected to Grover operators and the lower dimension of IAM operators fullfil the requirements of a operator $S$
\begin{equation}
\begin{split}
\big(G(s_1,s_2)-I(s_1,s_2)\big)(a\ket{s_1}+b\ket{s_2}) = O(1/\sqrt{N})(a\ket{s_1}+b\ket{s_2}), \\
\big(\IAM(s_1,s_2)-I(s_1,s_2)\big)(a\ket{s_1}+b\ket{s_2}) = a\ket{s_1}+O(1/\sqrt{N})b\ket{s_2}.
\end{split}
\end{equation}
The conditions of operators $A,B$ of Theorem~\ref{th:replace_theorem} can be applied for operators
\begin{equation}
\begin{split}
A = \IAM(s_1,s_2) \text{ and } B = I_1(s_1) \\
A = G(s_1,s_2)I_1(s_2) \text{ and } B = I_1(s_2)
\end{split}
\label{eq:grover_iam_approx_i1}
\end{equation}
for arbitrary states $s_1,s_2\in T_k$.

Furthermore, we can apply the Theorem~\,\ref{th:replace_theorem} on cubic structures of IAM operators like the stucture in $PG_3$ between the states $\ket{N_1S_2S_3}$.
We define a cubic structure recursively
\begin{align*}
&\IAM(gS_ih) = \IAM(ge_ig,sN_ih), \\
&\IAM(gT_{i,j}h) = [\IAM(gT_{i,j-1}e_{j}h)\IAM(gT_{i,j-1}N_{j}h)]\\
&\hspace{10em} \IAM(gT_{i,j-1}e_jh, gT_{i,j-1}N_jh)
\nonumber
\end{align*}
for arbitrary $g\in T_{j+i+2,k}$ and $s\in T_{j-1}$.
We define the matrix representation of the cubic structure $\IAM(T_k) \doteq C_k$.
\begin{theorem}
For any $k$ there exist a composition of reflections $I_i$, such that $C_kv = I_iv + O(1/N)$ for any vector $v$ with $\|v\|=O(1/\sqrt{N})$ and $C_k^2 = I + O(1/N)$.
\label{th:cubic_iam}
\end{theorem}
\begin{proof} 
$\IAM(T_k)$ is a composition of IAM operators. For any vector $v$ with $\|v\|\leq 1$ it is $C_1 v = I_1v + O(1/\sqrt{N})$.
Thus, by replacing each $\IAM$ with $I_1$ we obtain $I_i$ with an error in $O(1/\sqrt{N})$ and the first part of the theorem follows.



For the second part we use the fact that the IAM operator can be approximated for states $\ket{T_{k-1}e_k},\ket{T_{k-1}N_k}\in W_k$ with the amplitude vector $\|(v_e, v_N)\|\leq1+O(1/\sqrt{N})$ with 
\[\IAM(T_{k-1}e_k,T_{k-1}N_k)(v_e\ket{\mathoverlinee{T_{k-1}e_k}}+v_{F}\ket{\mathoverlinee{T_{k-1}N_k}}) 
\doteq \begin{bmatrix}-v_e+\frac{2v_N}{\sqrt{N}}\\v_N+\frac{2v_e}{\sqrt{N}}\end{bmatrix} + \epsilon\]
for a $\epsilon\in O(1/N)$.
As mentioned before, for any $k$ there exist a $I_i$ such that $\IAM(T_{k-1}) = I_i+O(1/\sqrt{N})$ .
Let $v$ be the amplitude vector of an arbitrary state in $Q_k$.
We express with $v_{e_k}$ the subspace $\spn(\{\ket{se_k}|s\in T_{k-1}\})$ and 
with $v_{N_k}$ the subspace $\spn(\{\ket{sN_k}|s\in T_{k-1}\})$ of the vector $v$.
We will show that $C_k^2v=v+O(1/N)$. In each step we omit an error in $O(1/N)$.
First the operation $\IAM(T_{k-1}e_k,T_{k-1}N_k)$ on $v$ returns
\[v'_{e_k} =  -v_{e_k} + \frac2{\sqrt{N}}v_{N_k},\qquad v'_{N_k} =  v_{N_k} + \frac2{\sqrt{N}}v_{e_k} .\]
As the next step, in the operation $\IAM(T_{k-1}e_{k})\IAM(T_{k-1}N_k)$ the $\IAM$ operations can be approximated with $I_i$ for parts of the vector in $O(1/\sqrt{N})$
\[v''_{e_k} =  -C_{k-1}v_{e_k} + \frac2{\sqrt{N}}I_iv_{N_k},\qquad v''_{N_k} =  C_{k-1}v_{N_k} + \frac2{\sqrt{N}}I_iv_{e_k}.\] 
The operator $\IAM(T_{k-1}e_k,T_{k-1}N_k)$ is applied again
\begin{align*}
&v'''_{F_e} =  -(-C_{k-1}v_{e_k} + \frac2{\sqrt{N}}I_iv_{N_k}) + \frac2{\sqrt{N}}C_{k-1}v_{N_k},\\
&v'''_{N_k} =  C_{k-1}v_{N_k} + \frac2{\sqrt{N}}I_iv_{e_k} - \frac2{\sqrt{N}}C_{k-1}v_{e_k}.
\end{align*}
Then the operation $\IAM(T_{k-1}e_k)\IAM(T_{k-1}N_k)$ returns $v$ again by
applying the induction hypothesis $C_{k-1}^2v=v+O(1/N)$
\begin{align*}
&v''''_{F_e} = v_{e_k} -\frac2{\sqrt{N}}v_{N_k} + \frac2{\sqrt{N}}v_{N_k} = v_{e_k}, \\
&v''''_{N_k} =  v_{N_k} + \frac2{\sqrt{N}}v_{e_k} - \frac2{\sqrt{N}}v_{e_k} = v_{N_k}.
\end{align*}
\end{proof}
As a consequence we can approximate the IAM operators in $PG_k$ with independent $C_i$ operators.
Further, we show that $PG_k$ only consists of independent $C_i$ operators and Grover operators.
Firstly, we define $t_k=e_1\ldots e_k$ where $t_0$ is in empty string, then the parallel Grover operator can be expressed recursively 
\begin{equation}
\begin{split}
PG_{1}(S_1h) &= G(S_1h) \\
PG_{i}(T_{i}h) &= PG_{i-1}(T_{i-1}e_{i}h)PG_{i-1}(T_{i-1}N_{i}h) \\
&\prod_{g\in T_{i-1}\setminus\{t_{i-1}\}}\IAM(ge_ih,gN_ih)G(t_{i-1}e_ih,t_{i-1}N_ih)
\label{eq:recursive_pgk}
\end{split}
\end{equation}
for an arbitrary $h\in T_{i+1,k}$. 
Basically, the recursive rule connects two $PG_{k-1}$ systems with one Grover operator between both their sink states and the remaining bipartite connections with IAM operators.
\begin{corollary}
The IAM operators in $PG_k$ can be approximated with $I_i$ operators.
\end{corollary}
\begin{proof}
The IAM operators in $PG_k$ can be expressed as cubic structures:
\[\prod_{i=0}^{k-2} \IAM(t_iN_{i+1}T_{i+2,k}).\]
Applying the induction hypothesis on the recursive form of $PG_k$ in Equation~\ref{eq:recursive_pgk} the IAM operators in $PG_k$ are
\begin{multline}
\prod_{i=0}^{k-3}\IAM(t_iN_{i+1}T_{i+2,k-1}e_k)\IAM(t_iN_{i+1}T_{i+2,k-1}N_k) \\
\prod_{g\in T_{k-1}\setminus\{t_{k-1}\}}\IAM(ge_kh,gN_kh).
\nonumber
\end{multline}
We reformulate the second term 
\begin{multline}
\prod_{g\in T_{k-1}\setminus\{t_{k-1}\}}\IAM(ge_kh,gN_kh) = \\ 
\prod_{i=0}^{k-3} \IAM(t_iN_{i+1}T_{i+2,k-1}e_k,t_iN_{i+1}T_{i+2,k-1}N_k),
\nonumber
\end{multline}
thus the IAM operators in $PG_k$ become
\begin{multline}
\bigg( \prod_{i=0}^{k-3}\IAM(t_iN_{i+1}T_{i+2,k-1}e_k)\IAM(t_iN_{i+1}T_{i+2,k-1}N_k)\\
\IAM(t_iN_{i+1}T_{i+2,k-1}e_k,t_iN_{i+1}T_{i+2,k-1}N_k)\bigg) = \prod_{i=0}^{k-2}\IAM(t_iN_{i+1}T_{i+2,k}).
\nonumber
\end{multline}
We can further see that for all $i=0,\ldots, k-2$ the $\IAM(t_iN_{i+1}T_{i+2,k})$ operators are separated from each other.
As a consequence we can apply Theorem~\ref{th:replace_theorem} on these cubic compositions.
\end{proof}
This explains why we can approximate $PG_k$ well by squaring and removing terms in $O(1/N)$ like it was done in the Appendix~\ref{sec:appendix_a} for $PG_2$. By squaring $PG_k$ the $I_i$ operators vanish. 

Furthermore, for each increasing $k$ at least one additional Grover operator diverts the amplitude from the main path in $PG_k$.
This effect increases even with diversions into higher depth.
We propose a solution for $3$-ISP. However we are not aware of any generalization of the solution for higher $k$'s.
We leave a general efficient solution for the $k$-ISP as open problem.

\subsection{Generalization of sequential Grover}
We can use a solution for the ISP for any $k$ to solve the ISP for $mk$ with $m\in\mathbb{N}$. Let $\ket{s_k}$ be the source state and $\ket{t_k}$ the sink state.
Let $A(x_1\cdots x_k)$ be a solution for the iterated search problem for any $k$ with $A\ket{s_k} = \ket{t_k}$.
Then we can solve the $mk$-ISP with the circuit
\[\prod_{i=0}^{m-1} A(x_{1+(m-i)k}\cdots x_{k+(m-i)k}) \ket{s_{mk}} = \ket{t_{mk}},\]
which is the same mechanism the sequential Grover uses. For example using the results of $PG_2$, we can solve the $3$-ISP within $(1+\sqrt{2})\pi\sqrt{N}/4 \approx 1.9\sqrt{N}$ iterations.

\section{A solution to $3$-ISP using $PG_3$}
\label{sec:solution_to_3isp}
We define the different parts of $PG_3$ as it was done for $PG_2$
\begin{multline}
\PGseq{3}(x_1x_2x_3) = \\ \underbrace{\IAM(x_1)O_1(x_1)}_{= PG_{3_1}(x_1)} 
\underbrace{\IAM(x_2)\oracle_2(x_1x_2)}_{= PG_{3_2}(x_1x_2)} \underbrace{\IAM(x_3)\oracle_3(x_1x_2x_3)}_{= PG_{3_3}(x_1x_2x_3)} .
\nonumber
\end{multline}
For each operator of the form $PG_{3_{i_0}}PG_{3_{i_1}}$ with $i_0 < i_1$, or $PG_{3_{i_0}}$ we can construct a cirucit, which executes the operation in one iteration. One iteration includes one or more oracle operators executed in parallel and one or more IAM operators executed in parallel.
With $PG_{3_2}$ the state $\ket{N_1e_2e_3}$ can be reflected without affecting any state on the main path significantly.
The reflection of $\ket{N_1e_2e_3}$ has the effect that the edge $G(e_1e_2e_3,N_1e_2e_3)$ rotates the amplitude back to $\ket{e_1e_2e_3}$ until the amplitude of $\ket{N_1e_2e_3}$ is zero again. Therefore, we first look for constants $c_1$ and $c_2$, such that
\begin{multline}
PG_3^{c_2\sqrt{N}}PG_{3_2}PG_3^{c_1\sqrt{N}}\ket{\mathoverlinee{N_1N_2N_3}} =\\ 
a_1\ket{\mathoverlinee{e_1e_2e_3}} + a_2\ket{\mathoverlinee{e_1e_2N_3}} + a_3\ket{\mathoverlinee{e_1N_2N_3}}
 + O(1/\sqrt{N})\ket{\mathoverlinee{T_3}}
\nonumber
\end{multline}
for some constants $a_1,a_2,a_3$.
The remaining amplitude can be transfered with $[PG_{3_2}PG_{3_3}]$ to the sink state without any diversion like it is done in $PG_2$.
However, to prevent that the amplitude of $\ket{e_1e_2N_3}$ changes its phase before the amplitude of $\ket{e_1N_2N_3}$ is zero,
we have to redistribute the amplitude between the states $\ket{e_1e_2e_3}$ and $\ket{e_1e_2N_3}$, such that $PG_2$ would obtain a solution. 
More formally, find $a_1'$ and $a_2'$ such that there exist a constant $d$ with
\[PG_2^{d\sqrt{N}}\ket{\mathoverlinee{N_1N_2}} = a_1'\ket{\mathoverlinee{e_1e_2}} + a_2'\ket{\mathoverlinee{e_1N_2}} + a_3\ket{\mathoverlinee{N_1N_2}} + O(1/N)\ket{\mathoverlinee{T_2}}.\]
We can obtain $a_2'$ by using the explicit functions of the amplitude evolution of $PG_2$ in Equation~\ref{eq:parallel_grover_k2_amplitude_evolution}.
Further, if $a_2' > a_2$, then the above explained case will occur so we have to retransfer amplitude from $\ket{e_1e_2e_3}$ to $\ket{e_1e_2N_3}$ such that
\begin{multline}
(PG_{3_3}^{\dagger})^{c_3\sqrt{N}}\big(a_1\ket{\mathoverlinee{e_1e_2e_3}} + a_2\ket{\mathoverlinee{e_1e_2N_3}} + a_3\ket{\mathoverlinee{e_1N_2N_3}}\big) =\\ a_1'\ket{\mathoverlinee{e_1e_2e_3}} + a_2'\ket{\mathoverlinee{e_1e_2N_3}} + a_3\ket{\mathoverlinee{e_1N_2N_3}} + O(1/N)\ket{\mathoverlinee{T_3}}.
\nonumber
\end{multline}
The $PG_{3_2}$ operator could also be used for this case, but the edge operator $G(e_1e_2e_3,e_1N_2e_3)$ would divert amplitude from $\ket{e_1e_2e_3}$ to $\ket{e_1N_2e_3}$.
For the other case where $a_2' \leq a_2$ this step can be skipped.

Then we apply $[PG_{3_2}PG_{3_3}]$ until the amplitude of state $\ket{e_1N_2N_3}$ is zero
\begin{multline}
[PG_{3_2}PG_{3_3}]^{c_4\sqrt{N}}\big( a'_1\ket{\mathoverlinee{e_1e_2e_3}} + a'_2\ket{\mathoverlinee{e_1e_2N_3}} + a_3\ket{\mathoverlinee{e_1N_2N_3}}\big) =\\ a_1''\ket{\mathoverlinee{e_1e_2e_3}} + a_2''\ket{\mathoverlinee{e_1e_2N_3}} + O(1/N)\ket{\mathoverlinee{T_3}}.
\nonumber
\end{multline}
If in the previous step $a_2' \geq a_2$, then the solution with negligible error is obtained, otherwise the remaining amplitude has to be transfered to the sink state with $PG_{3_3}^{c_5\sqrt{N}}$.

The calculations for the constants can be found in \cite{code_results} and are $c_1\approx0.78, c_2\approx0.17, c_3\approx0.05, c_4\approx0.5, c_5=0$ and a total number of iterations $\approx1.51\sqrt{N}$. The amplitude evolution on the operator can be seen in Figure~\ref{fig:k3_efficient_solution}.
\begin{figure}
\center\includegraphics[width=0.60\textwidth]{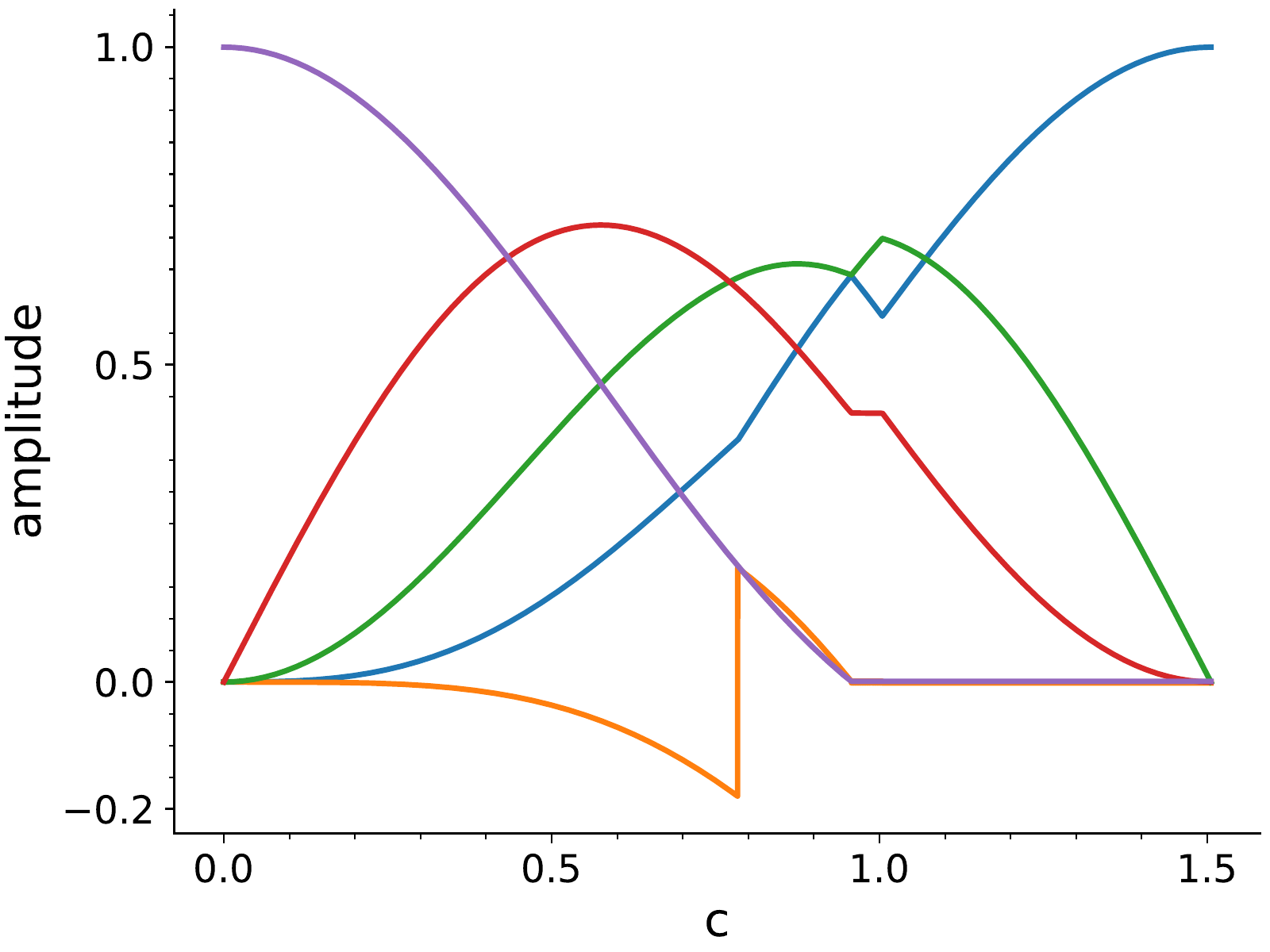}
\caption{The amplitude evolution of the solution for the $3$-ISP. \mbox{\raisebox{0.25ex}{\legendmarker{matplotlibgreenk31}} $\ket{\mathoverlineplot{e$_1$e$_2$e$_3$}}$}, \mbox{\raisebox{0.25ex}{\legendmarker{matplotlibgreenk32}} $\ket{\mathoverlineplot{N$_1$e$_2$e$_3$}}$}, \mbox{\raisebox{0.25ex}{\legendmarker{matplotlibgreenk33}} $\ket{\mathoverlineplot{e$_1$e$_2$N$_3$}}$}, \mbox{\raisebox{0.25ex}{\legendmarker{matplotlibgreenk34}} $\ket{\mathoverlineplot{e$_1$N$_2$N$_3$}}$}, \mbox{\raisebox{0.25ex}{\legendmarker{matplotlibgreenk35}} $\ket{\mathoverlineplot{N$_1$N$_2$N$_3$}}$}.}
\label{fig:k3_efficient_solution}
\end{figure}

\section{Lower Bound for $k$-ISP}
It can be concluded that the approximations of $PG_k$ only contains one path connecting the source state with the sink state.
Therefore, the number of iterations to reach the solution state with negligible error for circuits of the sequential form 
\[\prod_{l=1}^{O(\sqrt{N})}\IAM(x_{j_l})O_{i_l}(x_1\ldots x_{i_l})\]
for some $i_l, j_l\in\{1,\ldots, k\}$ for all $l$ can be lower bounded with the approximation for $PG_k$. The Grover operators, which divert the amplitude from any state on the main path, cannot speed up the transport of amplitude to the sink state. By removing these Grover operators only one path of Grover operators connecting the source with the sink state remains. Thus, a lower bound for a path of $k+1$ Grover operators can be used as lower bound for the $k$-ISP. A path of Grover operators can approximated with
\[P_k =
\begin{bmatrix}
          1 & 2/\sqrt{N} &             & 0 \\
-2/\sqrt{N} &     \ddots &      \ddots &   \\
            &     \ddots &      \ddots & 2/\sqrt{N}\\
          0 &            & -2/\sqrt{N} & 1 \end{bmatrix}\in\mathbb{R}^{k+1,k+1}.
\] 
By removing the $-2/\sqrt{N}$ terms we can lower bound the number of iterations. Let $u_i$ be the euclidean unit vector in the $i$th dimension and $\tilde{P}_k$ be the approximation of $P_k$ without the $-2/\sqrt{N}$ terms. 
Then we determine the $i$th power of the matrix $\tilde{P}_k$ in Jordan normal form and come to the conclusion that
\begin{equation}
\tilde{P}_{k}^{c\sqrt{N}}u_{k+1} = u_1\text{ for }c = \frac{(k!)^{\frac1k}}{2} + O(\frac1{\sqrt{N}})\geq \frac{k}{2e}.
\label{eq:lower_bound}
\end{equation}
The same approximation was done by Ozhigov \cite{ozhigov1999speedup}, but he did not determine an exact value.

\section{Summary}
In this paper, we gave an interpretation of $PG_2$ as rotation of a vector within a $3$-sphere and an interpretation of a class of solutions for the ISP as dynamic system system of IAM and Grover operators. Additionally, we gave an approximation method for $PG_k$, which explains the behaviour of the amplitude evolution. With the approximation method we showed that the sole application of $PG_k$ does not work as solution for the $k$-ISP for $k>2$. For $3$-ISP we presented a solution, which is more efficient than the general sequential Grover using $PG_2$. A solution for the diverted amplitude problem for $k$-ISP for $k>3$ remains as open problem.
\appendix

\section{\label{sec:appendix_a}Alternative proof for the amplitude evolution of $PG_2\ket{\mathoverline{N_1N_2}}$}
For an orthogonal matrix $A$, for large enough $N$ and constant $c$ it is
\begin{equation}
(A+O(1/N))^{c\sqrt{N}} = A^{c\sqrt{N}} + O(1/\sqrt{N}), \label{eq:approx_nthpower}
\end{equation}
 where $O(1/N)$ is applied componentwise.

As we have seen in equation \ref{eq:PG_seq_2_matrix_representation}, we can express one iteration of the $PG_2$ as
\begin{multline}
PG_2 = [G(e_1e_2, N_1e_2)G(e_1N_2, N_1N_2)] \\ [G(e_1e_2,e_1N_2)\IAM(N_1e_2,N_1N_2)].
\nonumber
\end{multline}
By removing the $O(1/N)$ terms we can approximate $PG_2^2$ using \ref{eq:approx_nthpower} with
\begin{equation}
PG_2^2 = [G(e_1N_2,N_1N_2)G(e_1e_2,e_1N_2)]^2.
\end{equation}
By further removing $O(1/N)$ terms, we can approximate this operation with the matrix
\[
P_2^2 = \begin{bmatrix}
1 & \frac4{\sqrt{N}} & 0 \\
-\frac{4}{\sqrt{N}} & 1 & \frac{4}{\sqrt{N}} \\
0 & -\frac{4}{\sqrt{N}} & 1 \end{bmatrix} 
\]
acting on $\ket{\mathoverline{e_1e_2}}, \ket{\mathoverline{e_1N_2}},\ket{\mathoverline{N_1N_2}}$. For tridiagonal Toeplitz matrices there exist a closed formula for the eigenvalues and eigenvectors \cite{noschese2013tridiagonal}. The inital state is 
\[PG_2^{c\sqrt{N}}H^{\otimes n}\ket{0} = PG_2^{c\sqrt{N}}\ket{\mathoverline{N_1N_2}} + O(1/\sqrt{N})\ket{\mathoverline{T_2}}.\]
Applying the eigendecomposition on $P_2^2$ the approximation of the initial amplitude vector we can approximate
\begin{multline}
PG_2^{c\sqrt{N}}\ket{\mathoverline{N_1N_2}} = \sin(\sqrt{2}c)^2\ket{\mathoverline{e_1e_2}} + \sqrt{2}\sin(\sqrt{2}c)\cos(\sqrt{2}c)\ket{\mathoverline{e_1N_2}}\\ + \cos(\sqrt{2}c)^2 \ket{\mathoverline{N_1N_2}}+ O(1/\sqrt{N})\ket{\mathoverline{T_2}},
\label{eq:parallel_grover_k2_amplitude_evolution}
\end{multline}
which agrees with Ozhigov's results \cite{ozhigov1999speedup}.

\end{document}